\documentclass[11pt,a4paper]{article}

\usepackage[T1]{fontenc}
\usepackage[utf8]{inputenc}
\usepackage{lmodern}
\usepackage[margin=1in]{geometry}

\usepackage{graphicx}%
\usepackage{amsmath,amssymb,amsfonts,mathtools}%
\usepackage{amsthm}%
\usepackage{mathrsfs}%
\usepackage{textcomp}%
\usepackage{microtype}%
\usepackage[authoryear,round]{natbib}
\usepackage[hidelinks]{hyperref}

\theoremstyle{plain}%
\newtheorem{theorem}{Theorem}%
\newtheorem{lemma}[theorem]{Lemma}%
\newtheorem{proposition}[theorem]{Proposition}%
\newtheorem{corollary}[theorem]{Corollary}%

\theoremstyle{remark}%
\newtheorem{remark}[theorem]{Remark}%

\newcommand{\R}{\mathbb R}
\newcommand{\cM}{\mathcal M}
\newcommand{\cF}{\mathcal F}
\newcommand{\dd}{\,d}
\newcommand{\KL}{D_{\mathrm{KL}}}
\newcommand{\one}{\mathbf 1}
\DeclareMathOperator*{\argmin}{arg\,min}

\begin{document}

\title{Numeraire Invariance of Entropy-Projected Martingale Measures}
\author{Jan Vecer\thanks{Charles University, Sokolovska 83, Prague 8, 18675, Czech Republic. Email: \texttt{vecer@karlin.mff.cuni.cz}.}}
\date{}

\maketitle

\begin{abstract}
Let \(P\) be a fixed physical law and let \(Q\) be an equivalent martingale measure selected from the martingale-measure set associated with a chosen numeraire. A change of numeraire maps \(Q\) to \(T_LQ\), where \(d(T_LQ)=L\,dQ\) and \(L\) is the terminal likelihood ratio. The forward relative-entropy projection minimizing \(D_{\mathrm{KL}}(P\Vert Q)\) commutes with this transform because its objective changes only by the constant \(-E_P\log L\). The
minimal entropy martingale measure
(MEMM) orientation \(D_{\mathrm{KL}}(Q\Vert P)\) does not have this property, and a trinomial counterexample shows that independently recomputed MEMMs need not be likelihood compatible. We make two economic consequences explicit. First, the two entropy orientations are precisely the \(Q\)-dependent terms in the classical convex-dual objectives for logarithmic and exponential utility, respectively. Second, likelihood compatibility is equivalent to equality of the pricing functionals obtained in the two numeraires. Hence the forward selectors value every integrable claim consistently across numeraires, whereas the two MEMMs in the counterexample assign different prices to a nonreplicable digital claim. We also prove a finite-state class-level characterization: uniform invariance over the elementary one-period likelihood-ratio families forces a smooth convex \(f\)-divergence to be logarithmic, up to scaling and affine equivalence. Finally, in finite-state markets, the forward projection exists under the usual strictly positive feasible-point condition; its density \(dP/dQ^*\) is attainable log-optimal terminal wealth, and the minimum forward entropy equals maximal expected log growth.
\end{abstract}



\section{Introduction}\label{sec:introduction}

Equivalent martingale measures are generally not unique in incomplete markets. This non-uniqueness motivates selection rules, among them minimal-distance martingale measures and, in particular, the minimal entropy martingale measure (MEMM). The MEMM minimizes
\[
        \KL(Q\Vert P)=E_Q\!\left[\log\frac{\dd Q}{\dd P}\right]
\]
over equivalent martingale measures \(Q\), and is closely connected with exponential utility and utility indifference valuation; see, for example, \citet{Frittelli2000}, \citet{DelbaenEtAl2002}, \citet{GranditsRheinlaender2002}, and the survey by \citet{SchweizerMEMM}. More generally, \citet{GollRueschendorf2001} characterize minimal \(f\)-divergence martingale measures and relate them to utility maximization.

This paper studies the opposite orientation of the same relative entropy. Given a physical law \(P\), consider
\[
        Q^*\in\argmin_{Q\in\cM} \KL(P\Vert Q),
        \qquad
        \KL(P\Vert Q)=E_P\!\left[\log\frac{\dd P}{\dd Q}\right],
\]
where \(\cM\) is the martingale-measure set for a given numeraire. In the terminology of Goll and R\"uschendorf, this is the logarithmic case \(f(x)=-\log x\); in utility-duality language it is the dual martingale measure associated with logarithmic utility. The logarithmic construction is known. The present paper isolates a structural property of it---compatibility with likelihood-ratio changes of numeraire---and compares that property with the MEMM orientation.

The two orientations have parallel optimization origins. If \(Z_Q=dQ/dP\), their \(Q\)-dependent dual terms are
\[
 -E_P[\log Z_Q]=\KL(P\Vert Q)
 \quad\text{and}\quad
 E_P[Z_Q\log Z_Q]=\KL(Q\Vert P)
\]
for logarithmic and exponential utility, respectively. The distinction is therefore not between an optimization-based selector and an ad hoc one; it is between two utility duals with different behavior under changes of accounting unit. Section~\ref{subsec:utility_duality} gives the conjugate calculations explicitly.

Numeraire compatibility also has a direct pricing meaning. Let \(Q^N\) and \(Q^X\) be measures selected in economies denominated by two strictly positive numeraires \(N\) and \(X\). The standard pricing formulas in the two units agree for every integrable contingent claim if and only if \(Q^X\) is the likelihood transform of \(Q^N\). Thus two agents who share the same physical law and traded assets, but independently solve a selector in their respective numeraires, can obtain different prices for nonreplicable derivatives if their selected measures are not likelihood compatible. They continue to agree on traded and replicable claims. The forward projection avoids this additional source of disagreement; independently recomputed MEMMs need not.

The contribution is fivefold. First, Theorem~\ref{thm:forward_invariance} records the constant-shift identity in extended-real form, and Corollary~\ref{cor:entropy_log_return} rewrites it as an entropy decomposition of the expected log return of one numeraire relative to another. Second, Proposition~\ref{prop:pricing_consistency} shows that likelihood compatibility is equivalent to cross-numeraire pricing consistency, while Section~\ref{subsec:utility_duality} identifies the two entropy orientations through utility conjugacy. Third, Theorem~\ref{thm:f_divergence_uniqueness} gives a finite-state class-level characterization: among smooth \(f\)-divergence projections, the logarithmic integrand is the unique nontrivial integrand whose minimizers are invariant under the whole class of elementary one-period numeraire-compatible likelihood transforms. The class-level quantifier is essential; a single finite market tests only the finitely many values taken by its likelihood ratio. Fourth, Lemma~\ref{lem:finite_existence} and Theorem~\ref{thm:forward_log_wealth} separate finite-state existence from first-order conditions and establish the exact primal-dual identity between minimum forward entropy and maximum expected logarithmic growth. Fifth, Theorem~\ref{thm:reverse_counterexample} gives an explicit trinomial failure of MEMM invariance, and Corollary~\ref{cor:memm_price_disagreement} exhibits the resulting disagreement on the price of a nonreplicable digital claim.

Finite-state arguments are used only where existence and first-order conditions are asserted. The change-of-numeraire identity and the pricing-consistency equivalence are measure-theoretic. In continuous-time incomplete semimartingale models, however, the logarithmic dual optimizer need not be an equivalent martingale measure and may live in an enlarged deflator domain, as in the general utility-maximization theory of \citet{KramkovSchachermayer1999} and the numeraire-portfolio theory of \citet{Long1990}, \citet{Becherer2001}, and \citet{KaratzasKardaras2007}. Section~\ref{sec:scope} returns to this point.

\section{Entropy projections and numeraire transforms}\label{sec:entropy_transforms}

Let \((\Omega,\cF)\) be a measurable space and let \(P\) be a fixed physical probability measure. For probability measures \(\mu,\nu\), define
\[
        \KL(\mu\Vert \nu)
        :=
        \begin{cases}
        E_\mu\!\left[\log\dfrac{d\mu}{d\nu}\right], & \mu\ll\nu,\\[2mm]
        +\infty, & \text{otherwise.}
        \end{cases}
\]
We call \(\KL(P\Vert Q)\) the \emph{forward} orientation because the expectation is under the fixed physical law \(P\). We call \(\KL(Q\Vert P)\) the \emph{reverse} orientation. The labels are only mnemonic; all formulas display the order of the arguments explicitly. This avoids a terminology collision with parts of the martingale-measure literature, where \(\KL(P\Vert Q)\) is sometimes described as a reverse relative entropy because it corresponds to the divergence function \(f(x)=-\log x\).

Suppose that \(N\) and \(X\) are strictly positive traded numeraires over a finite horizon \([0,T]\). Under the usual change-of-numeraire hypotheses, if \(Q\) is an equivalent martingale measure in units of \(N\), then the corresponding measure in units of \(X\) is obtained by a fixed likelihood-ratio transform \citep{GemanElKarouiRochet1995}. We write \(T_L\) for this transform on measures: \(T_LQ\) denotes the probability measure obtained from \(Q\) by multiplication with the likelihood ratio \(L\). Equivalently,
\begin{equation}\label{eq:change_numeraire_density}
        \frac{d(T_LQ)}{dQ}
        =L,
        \qquad
        L:=\frac{X_T/N_T}{X_0/N_0}.
\end{equation}
Here \(L>0\) is fixed by the terminal price ratio and \(E_Q[L]=1\) for each admissible \(Q\). In a one-period finite-state model, \(T_LQ=LQ\) is simply multiplication of state probabilities by the gross return of the new numeraire measured in old-numeraire units. When the likelihood ratio is denoted by another symbol, for example \(R\), we use the corresponding notation \(T_R\).

The following theorem is the basic observation.

\begin{theorem}[Forward entropy projection commutes with numeraire change]\label{thm:forward_invariance}
Let \(\mathcal C\) be a family of probability measures on \((\Omega,\cF)\). Let \(L:\Omega\to(0,\infty)\) be measurable, suppose that \(E_Q[L]=1\) for every \(Q\in\mathcal C\), and define \(T_LQ\) by \(d(T_LQ)/dQ=L\). If \(E_P|\log L|<\infty\), then, for every \(Q\in\mathcal C\),
\begin{equation}\label{eq:forward_constant_shift}
        \KL(P\Vert T_LQ)
        =
        \KL(P\Vert Q)-E_P[\log L]
\end{equation}
in the extended-real sense, with the convention \(+\infty-c=+\infty\) for finite \(c\). Consequently,
\[
        T_L\left(\argmin_{Q\in\mathcal C}\KL(P\Vert Q)\right)
        =
        \argmin_{\widetilde Q\in T_L(\mathcal C)}\KL(P\Vert\widetilde Q),
\]
whenever the argmin sets are nonempty.
\end{theorem}

\begin{proof}
Since \(L\) is strictly positive and finite, \(Q\) and \(T_LQ\) have the same null sets. Hence, if \(P\not\ll Q\), then also \(P\not\ll T_LQ\), and both sides of \eqref{eq:forward_constant_shift} equal \(+\infty\). If \(P\ll Q\), then
\[
        \frac{\dd P}{\dd(T_LQ)}=\frac{dP/dQ}{L}.
\]
Therefore
\[
        \KL(P\Vert T_LQ)
        =E_P\!\left[\log\frac{\dd P}{\dd Q}-\log L\right]
        =\KL(P\Vert Q)-E_P[\log L],
\]
where the equality is understood in the extended-real sense; the second term is finite by assumption and is independent of \(Q\). Thus the transform changes the objective by an additive constant and preserves minimizers.\end{proof}

\begin{corollary}[Numeraire invariance of the forward martingale-measure projection]\label{cor:emm_forward_invariance}
Let \(\cM^N\) and \(\cM^X\) be the equivalent martingale-measure sets associated with numeraires \(N\) and \(X\), and suppose that the change-of-numeraire map \(T_L\) in \eqref{eq:change_numeraire_density} is a bijection from \(\cM^N\) to \(\cM^X\), and assume \(E_P|\log L|<\infty\). Then
\[
        Q_N^*\in\argmin_{Q\in\cM^N}\KL(P\Vert Q)
        \quad\Longleftrightarrow\quad
        T_LQ_N^*\in\argmin_{\widetilde Q\in\cM^X}\KL(P\Vert\widetilde Q).
\]
\end{corollary}

\begin{corollary}[Expected log return as an entropy difference]\label{cor:entropy_log_return}
Let \(Q^N\in\cM^N\) and let \(Q^X=T_LQ^N\), with \(L\) as in \eqref{eq:change_numeraire_density}. If \(E_P|\log L|<\infty\) and one, hence both, of the divergences below is finite, then
\begin{equation}\label{eq:entropy_log_return}
 E_P\!\left[\log\frac{X_T/N_T}{X_0/N_0}\right]
 =\KL(P\Vert Q^N)-\KL(P\Vert Q^X).
\end{equation}
\end{corollary}

\begin{proof}
This is a rearrangement of \eqref{eq:forward_constant_shift}. The left-hand side is the expected logarithmic return of the \(X\)-numeraire relative to the \(N\)-numeraire under the fixed physical law.
\end{proof}

\begin{proposition}[Likelihood compatibility and cross-numeraire pricing]\label{prop:pricing_consistency}
Let \(Q^N\in\cM^N\) and \(Q^X\in\cM^X\). For a bounded \(\cF\)-measurable random variable \(G\), consider the terminal claim \(H_T=N_TG\) and the two pricing functionals
\begin{equation}\label{eq:two_pricing_functionals}
 \pi_N(H):=N_0E_{Q^N}\!\left[\frac{H_T}{N_T}\right],
 \qquad
 \pi_X(H):=X_0E_{Q^X}\!\left[\frac{H_T}{X_T}\right].
\end{equation}
Then
\[
 \pi_N(N_TG)=\pi_X(N_TG)\quad\text{for every bounded }G
 \quad\Longleftrightarrow\quad
 Q^X=T_LQ^N.
\]
Consequently, likelihood compatibility is exactly the condition under which the two numeraires induce the same linear valuation on contingent claims.
\end{proposition}

\begin{proof}
Since
\[
 \frac{N_T}{X_T}=\frac{N_0}{X_0}\frac1L,
\]
we have
\[
 \pi_X(N_TG)=N_0E_{Q^X}\!\left[\frac{G}{L}\right].
\]
If \(Q^X=T_LQ^N\), this equals \(N_0E_{Q^N}[G]=\pi_N(N_TG)\). Conversely, equality for every bounded \(G\) implies
\[
 E_{Q^N}[G]=E_{Q^X}\!\left[\frac{G}{L}\right]
 \quad\text{for every bounded }G,
\]
so \(dQ^N=L^{-1}dQ^X\), equivalently \(dQ^X=L\,dQ^N\).
\end{proof}

Combining Proposition~\ref{prop:pricing_consistency} with Corollary~\ref{cor:emm_forward_invariance}, a forward minimizer and its transformed counterpart define one pricing functional expressed in two different units: for every claim with integrable discounted payoffs,
\begin{equation}\label{eq:forward_price_consistency}
 N_0E_{Q_N^*}\!\left[\frac{H_T}{N_T}\right]
 =X_0E_{T_LQ_N^*}\!\left[\frac{H_T}{X_T}\right].
\end{equation}
On the convex set of measures equivalent to \(P\), the forward objective is strictly convex whenever it is finite, so an existing minimizer is unique; independently solved forward problems therefore recover this compatible pair. If two independently selected measures are not likelihood compatible, the converse part of the proposition supplies an event \(A\) for which the digital claim \(N_T\one_A\) receives different prices. Because both measures are martingale measures for the same traded assets, such disagreement can occur only on nonreplicable claims.

\subsection{Utility-dual origin of the two entropy orientations}\label{subsec:utility_duality}

The opposite entropy orientations arise from two standard utility conjugates. Let
\[
 Z_Q:=\frac{\dd Q}{\dd P},
 \qquad
 V(y):=\sup_x\{U(x)-xy\},\quad y>0.
\]
At the level of the classical martingale-measure dual problem, without random endowment, the \(Q\)-dependent term is \(E_P[V(yZ_Q)]\). For logarithmic utility \(U_{\log}(x)=\log x\), \(x>0\),
\[
 V_{\log}(y)=-\log y-1,
\]
and therefore
\begin{equation}\label{eq:log_conjugate_entropy}
 E_P[V_{\log}(yZ_Q)]
 =-\log y-1-E_P[\log Z_Q]
 =-\log y-1+\KL(P\Vert Q).
\end{equation}
Thus minimization over \(Q\) selects the forward entropy projection. For exponential utility \(U_{\exp}(x)=-e^{-\gamma x}\), \(x\in\R\), \(\gamma>0\),
\[
 V_{\exp}(y)=\frac{y}{\gamma}\left(\log\frac{y}{\gamma}-1\right),
\]
and, using \(E_P[Z_Q]=1\),
\begin{equation}\label{eq:exp_conjugate_entropy}
 E_P[V_{\exp}(yZ_Q)]
 =\frac{y}{\gamma}\left(
       \KL(Q\Vert P)+\log\frac{y}{\gamma}-1
   \right).
\end{equation}
The \(Q\)-minimizer is therefore the MEMM. Equations~\eqref{eq:log_conjugate_entropy}--\eqref{eq:exp_conjugate_entropy} give the precise sense in which the forward projection is tied to logarithmic utility while the MEMM is tied to exponential utility; see \citet{Frittelli2000}, \citet{GollRueschendorf2001}, and \citet{GollKallsen2000,GollKallsen2003}.

The same distinction appears on the primal side. For an actual wealth process \(V\), define the normalized terminal gross wealths
\[
 W_T^N:=\frac{V_T/N_T}{V_0/N_0},
 \qquad
 W_T^X:=\frac{V_T/X_T}{V_0/X_0}.
\]
Then \(W_T^X=W_T^N/L\), and hence
\begin{equation}\label{eq:log_utility_numeraire_shift}
 E_P[\log W_T^X]
 =E_P[\log W_T^N]-E_P[\log L],
\end{equation}
so the maximizing portfolio is unchanged by the representation change. Exponential utility does not transform by an additive portfolio-independent constant: in \(X\)-units its criterion contains \(-\exp(-\gamma W_T^N/L)\). This is consistent with, rather than a defect of, the MEMM's role in exponential-utility valuation.

\begin{remark}[Accounting-unit dependence and numeraire invariance]
The physical law \(P\) is held fixed throughout. If a selector is intended to provide a canonical state-price completion of an incomplete market, covariance under a pure change of numeraire is a natural representation requirement, and Proposition~\ref{prop:pricing_consistency} shows that it has observable pricing content. If instead the measure is introduced as the dual optimizer for exponential utility specified in a fixed currency or asset unit, dependence on that unit is part of the preference specification. The two criteria answer different economic questions.
\end{remark}

\section{A uniqueness characterization among smooth \texorpdfstring{\(f\)}{f}-divergences}\label{sec:f_uniqueness}

The constant-shift identity in Theorem~\ref{thm:forward_invariance} is not merely a formal algebraic identity. Within the usual finite-state \(f\)-divergence class it characterizes the logarithmic integrand, up to the affine changes of integrand that do not change an \(f\)-divergence minimization problem.

Let \(f:(0,\infty)\to\R\) be convex and continuously differentiable. For a measure \(Q\ll P\), write
\[
        z_i:=\frac{Q(\{\omega_i\})}{p_i}
        \quad\text{in finite state models,}
        \qquad
        I_f(Q\Vert P):=E_P\left[f\left(\frac{\dd Q}{\dd P}\right)\right].
\]
Thus \(z_i\) denotes the \(P\)-density of \(Q\), not the mass of \(Q\) at \(\omega_i\). The forward entropy \(\KL(P\Vert Q)\) is the case \(f(x)=-\log x\). More generally, replacing \(f\) by
\[
        \widetilde f(x)=af(x)+b(x-1)+c,
        \qquad a>0,
\]
multiplies the objective by a positive constant and adds constants on every probability family, so it leaves all argmin sets unchanged. The following theorem is stated modulo this standard affine equivalence.

\begin{theorem}[Class-level characterization of the logarithmic \(f\)-divergence]\label{thm:f_divergence_uniqueness}
Let \(f\in C^1(0,\infty)\) be convex. Suppose that the following pulled-back numeraire-invariance property holds for every finite one-period likelihood-ratio family: for every finite physical law \(P=(p_i)\) with \(p_i>0\), every strictly positive likelihood ratio \(L=(L_i)\) with \(E_P[L]=1\), and the elementary numeraire-compatible family
\[
        \mathcal C_L
        :=
        \left\{Q:\ z_i>0,\ \sum_i p_iz_i=1,\ \sum_i p_iL_iz_i=1\right\},
        \qquad z_i=\frac{Q(\{\omega_i\})}{p_i},
\]
the point \(z_i\equiv1\), which minimizes \(I_f(\cdot\Vert P)\) over \(\mathcal C_L\), also minimizes the pulled-back transformed objective
\[
        Q\longmapsto I_f(T_LQ\Vert P)=\sum_i p_i f(L_i z_i)
\]
over \(\mathcal C_L\). Then
\[
        f(x)=\alpha\log x+\beta x+\gamma,
        \qquad x>0,
\]
for constants \(\alpha,\beta,\gamma\), with \(\alpha\le0\) by convexity. If the divergence is non-affine and convex, then \(\alpha<0\); up to affine equivalence and positive scaling,
\[
        f(x)=-\log x.
\]
Conversely, \(f(x)=-c\log x+\beta x+\gamma\), \(c>0\), has the invariance property. On every numeraire-compatible family satisfying \(E_P[\dd Q/\dd P]=E_P[L\,\dd Q/\dd P]=1\), it has the stronger constant-shift property
\[
        I_f(T_LQ\Vert P)=I_f(Q\Vert P)-cE_P[\log L].
\]
\end{theorem}

\begin{proof}
The sufficiency follows from
\[
        E_P[-\log(Lz)]=E_P[-\log z]-E_P[\log L],
        \qquad z=\frac{\dd Q}{\dd P},
\]
and from the fact that, on \(\mathcal C_L\), the affine terms \(\beta x+\gamma\) contribute only constants because \(E_P[z]=E_P[Lz]=1\).

For necessity, put
\[
        g(x):=xf'(x).
\]
Fix a finite model and a positive \(L\) with \(E_P[L]=1\). Since \(f\) is convex, Jensen's inequality gives
\[
        I_f(Q\Vert P)\ge f(E_P[\dd Q/\dd P])=f(1),
\]
so the point \(Q=P\), i.e. \(z_i\equiv1\), minimizes \(I_f(Q\Vert P)\) over every feasible family containing it. By the assumed invariance property, the same pulled-back point must minimize
\[
        Q\mapsto I_f(T_LQ\Vert P)
        =\sum_i p_i f(L_iz_i)
\]
over \(\mathcal C_L\). The first-order condition at \(z_i\equiv1\), for variations \(h=(h_i)\) satisfying
\[
        \sum_i p_ih_i=0,
        \qquad
        \sum_i p_iL_ih_i=0,
\]
is
\[
        \sum_i p_i L_if'(L_i)h_i=0.
\]
Equivalently, the vector \(g(L_i)=L_if'(L_i)\) lies in the span of the two constraint vectors \(1\) and \(L_i\). Thus, for every finite list of positive numbers \(L_i\) for which \(1\) lies in the interior of their convex hull, the points
\[
        (L_i,g(L_i))
\]
lie on a single affine line.

This forces \(g\) itself to be affine on \((0,\infty)\). Indeed, given any three distinct positive numbers \(x,y,z\), choose a fourth positive number \(r\) so that \(1\) lies in the interior of the convex hull of \(\{x,y,z,r\}\). Then there exist strictly positive weights \(p_i\) whose weighted average of \(x,y,z,r\) is \(1\). The preceding paragraph implies that the four points \((x,g(x)),(y,g(y)),(z,g(z)),(r,g(r))\) are collinear; hence any three distinct points of the graph of \(g\) are collinear. Therefore
\[
        g(x)=A+Bx
\]
for constants \(A,B\). Since \(g(x)=xf'(x)\),
\[
        f'(x)=\frac{A}{x}+B,
\]
and integration gives
\[
        f(x)=A\log x+Bx+C.
\]
Convexity gives \(f''(x)=-A/x^2\ge0\), hence \(A\le0\). The non-affine case has \(A<0\), which is, up to affine equivalence and positive scaling, the integrand \(-\log x\).\end{proof}

\begin{remark}[Role of the class-level quantifier]
The quantifier in Theorem~\ref{thm:f_divergence_uniqueness} is essential. For one fixed finite vector \(L\), the first-order condition constrains \(xf'(x)\) only at the finitely many values taken by \(L\). The conclusion that \(xf'(x)\) is affine on all of \((0,\infty)\) follows from requiring the same pulled-back invariance property for the whole class of elementary one-period likelihood-ratio families. Thus the logarithmic divergence is forced before multi-period or continuous-time structure enters, but not by a single finite family alone.
\end{remark}

\begin{remark}
Theorem~\ref{thm:f_divergence_uniqueness} is deliberately finite-state and smooth. It gives a direct divergence-level version of the classical numeraire-invariance property of logarithmic utility. The necessity proof uses only the subclass of models in which the physical law \(P\) is itself feasible as a martingale measure; thus nonlogarithmic \(f\)-divergences already fail invariance before any nontrivial market price of risk is introduced. Within the Goll--R\"uschendorf minimal \(f\)-divergence framework, imposing numeraire invariance singles out the logarithmic case.
\end{remark}

\section{The finite-state forward projection and log-optimal wealth}\label{sec:forward_log}

We now specialize to a one-period finite-state market. In this section and the finite-state sections that follow, \(q_i\) denotes the mass \(Q(\{\omega_i\})\). Let \(\Omega=\{\omega_1,\ldots,\omega_m\}\), let \(P=(p_i)_{i=1}^m\) satisfy \(p_i>0\), and let the reference numeraire have return one. There are \(n\) further traded assets with strictly positive gross returns
\[
        R^j_i:=R^j(\omega_i),
        \qquad j=1,\ldots,n,
        \quad i=1,\ldots,m.
\]
A martingale measure \(Q=(q_i)\) satisfies
\begin{equation}\label{eq:finite_martingale_constraints}
        \sum_{i=1}^m q_i=1,
        \qquad
        \sum_{i=1}^m q_i R_i^j=1,
        \quad j=1,\ldots,n.
\end{equation}
Let \(\overline{\cM}\) denote the closed feasible polytope of nonnegative solutions of \eqref{eq:finite_martingale_constraints}, and let \(\cM\) denote its strictly positive part. The projection below is a finite-dimensional instance of Csisz\'ar's \(I\)-projection onto a linearly constrained family \citep{Csiszar1975}. For portfolio cost weights \(w=(w^0,w^1,\ldots,w^n)\) satisfying \(w^0+\sum_{j=1}^n w^j=1\), the terminal gross return is
\begin{equation}\label{eq:portfolio_return}
        R^w_i
        =w^0+\sum_{j=1}^n w^jR_i^j.
\end{equation}
We allow unconstrained long-short weights but require \(R^w_i>0\) in every state. For every \(Q\in\cM\) and every such portfolio, \(E_Q[R^w]=1\), so \(R^wQ\) is a probability measure and
\begin{equation}\label{eq:log_growth_bound}
 \KL(P\Vert Q)-E_P[\log R^w]
 =\KL(P\Vert R^wQ)\ge0.
\end{equation}
Consequently,
\begin{equation}\label{eq:log_weak_duality}
 \sup_{w:\,R^w>0}E_P[\log R^w]
 \le
 \inf_{Q\in\cM}\KL(P\Vert Q).
\end{equation}
Thus every martingale measure supplies an upper bound on achievable expected log growth, and the forward projection supplies the smallest such bound. The result below shows that this bound is attained in the finite-state interior case.

\begin{lemma}[Existence and interiority of the finite-state forward projection]\label{lem:finite_existence}
Assume \(\cM\neq\varnothing\). Then \(Q\mapsto\KL(P\Vert Q)\) attains its minimum over \(\overline{\cM}\). Every minimizer belongs to \(\cM\), and the minimizer is unique.
\end{lemma}

\begin{proof}
Extend
\[
        F(q)=\sum_{i=1}^m p_i\log\frac{p_i}{q_i}
\]
to \(\overline{\cM}\) by setting \(F(q)=+\infty\) if some \(q_i=0\). The set \(\overline{\cM}\) is a compact polytope, and the assumption \(\cM\neq\varnothing\) gives at least one feasible point with finite objective value. The extended function \(F\) is lower semicontinuous on \(\overline{\cM}\), so it attains its minimum. Since \(p_i>0\) for all \(i\), every boundary point with some \(q_i=0\) has infinite objective value; hence every minimizer is strictly positive. Strict convexity of \(q\mapsto -\sum_i p_i\log q_i\) on the positive orthant gives uniqueness on the convex feasible set.\end{proof}

\begin{theorem}[Forward projection density is log-optimal wealth]\label{thm:forward_log_wealth}
Assume \(\cM\neq\varnothing\), and let \(Q^*\) be the unique minimizer of \(\KL(P\Vert Q)\) over \(\overline{\cM}\). Then there exists an admissible portfolio \(w^*\) such that
\begin{equation}\label{eq:density_attainable}
        R_i^{w^*}=\frac{p_i}{q_i^*},
        \qquad i=1,\ldots,m.
\end{equation}
Consequently,
\[
        R^{w^*}Q^*=P,
        \qquad
        \KL(P\Vert R^{w^*}Q^*)=0,
\]
and \(w^*\) is growth optimal. Moreover, the primal and dual values coincide:
\begin{equation}\label{eq:log_primal_dual_equality}
 \max_{w:\,R^w>0}E_P[\log R^w]
 =\min_{Q\in\overline{\cM}}\KL(P\Vert Q)
 =\KL(P\Vert Q^*).
\end{equation}
\end{theorem}

\begin{proof}
By Lemma~\ref{lem:finite_existence}, \(Q^*\) is strictly positive, so only the equality constraints in \eqref{eq:finite_martingale_constraints} are active. The first-order condition for minimizing
\[
        \sum_{i=1}^m p_i\log\frac{p_i}{q_i}
\]
subject to \eqref{eq:finite_martingale_constraints} says that the gradient vector lies in the span of the constraint vectors. Equivalently, after deleting redundant constraints if necessary, there exist constants \(\lambda_0,\lambda_1,\ldots,\lambda_n\) such that
\[
        -\frac{p_i}{q_i^*}
        +\lambda_0+\sum_{j=1}^n\lambda_jR_i^j=0,
        \qquad i=1,\ldots,m.
\]
Thus
\[
        \frac{p_i}{q_i^*}
        =\lambda_0+\sum_{j=1}^n\lambda_jR_i^j.
\]
Taking expectation under \(Q^*\) and using the martingale constraints gives
\[
        1=\sum_i q_i^*\frac{p_i}{q_i^*}
        =\lambda_0+\sum_{j=1}^n\lambda_j.
\]
Therefore \(w^{0,*}=\lambda_0\) and \(w^{j,*}=\lambda_j\) are cost weights summing to one, and \eqref{eq:density_attainable} holds. Positivity follows from \(p_i/q_i^*>0\).

For any admissible \(w\), the measure with masses \(R_i^wq_i^*\) is a probability measure, because \(E_{Q^*}[R^w]=1\), and
\[
        \KL(P\Vert R^wQ^*)
        =
        \KL(P\Vert Q^*)-E_P[\log R^w].
\]
The left-hand side is nonnegative for every \(w\) and equals zero at \(w=w^*\). Therefore
\[
 E_P[\log R^w]\le \KL(P\Vert Q^*)
 =E_P[\log R^{w^*}],
\]
which proves growth optimality and \eqref{eq:log_primal_dual_equality}.\end{proof}

\begin{remark}[The sharpest log-growth bound]
Equation~\eqref{eq:log_growth_bound} gives a useful interpretation of the dual problem. Each martingale measure \(Q\) imposes the bound \(E_P\log R^w\le\KL(P\Vert Q)\) on every positive attainable return. Minimizing the forward divergence chooses the sharpest of these bounds, and Theorem~\ref{thm:forward_log_wealth} shows that its first-order condition makes the bound exact. Equivalently, the dual density is the reciprocal of optimal wealth, or \(dP/dQ^*\) is the log-optimal payoff; see \citet{GollKallsen2000,GollKallsen2003}.
\end{remark}

\section{Reverse entropy does not commute with numeraire change}\label{sec:reverse}

In this section, \(q_i\) denotes the mass \(Q(\{\omega_i\})\).

The MEMM orientation behaves differently because the expectation is taken under the measure being optimized.  If \(\widetilde Q=T_LQ\), then
\begin{align}\label{eq:reverse_difference_general}
        \KL(\widetilde Q\Vert P)-\KL(Q\Vert P)
        &= E_Q\!\left[L\log\left(L\frac{\dd Q}{\dd P}\right)\right]
           -E_Q\!\left[\log\frac{\dd Q}{\dd P}\right]  
           \notag\\
        &= E_Q\!\left[(L-1)\log\frac{\dd Q}{\dd P}\right]+E_Q[L\log L].
\end{align}
The right-hand side generally depends on \(Q\).  A nonconstant perturbation alone does not prove that minimizers fail to correspond, so we give an explicit counterexample.

\begin{theorem}[Reverse entropy is not numeraire invariant]\label{thm:reverse_counterexample}
There exists a one-period trinomial market and a physical law \(P\) such that the reverse-entropy minimizer in one numeraire does not map, under the change-of-numeraire likelihood transform, to the reverse-entropy minimizer in the other numeraire.
\end{theorem}

\begin{proof}
Let \(\Omega=\{\omega_1,\omega_2,\omega_3\}\), fix
\[
        P=\left(\frac12,\frac16,\frac13\right),
        \qquad
        R=\left(\frac32,1,\frac12\right),
\]
and interpret \(R\) as the gross return of asset \(X\) in units of asset \(Y\).  The \(Y\)-numeraire martingale measures are
\begin{equation}\label{eq:Y_family}
        Q^Y(a)=\left(\frac{1-a}{2},\,a,\,\frac{1-a}{2}\right),
        \qquad 0<a<1,
\end{equation}
which satisfy \(E_{Q^Y(a)}[R]=1\).  The change to the \(X\)-numeraire multiplies by \(R\):
\begin{equation}\label{eq:X_family}
        Q^X(a)=RQ^Y(a)
        =\left(\frac{3(1-a)}{4},\,a,\,\frac{1-a}{4}\right).
\end{equation}
These are exactly the martingale measures for the reciprocal return \(1/R=(2/3,1,2)\) in \(X\)-units.  Indeed, any measure in the same one-parameter form can be written as
\[
        \bigl((1-a)b,\,a,\,(1-a)c\bigr),
        \qquad b+c=1.
\]
The \(X\)-numeraire martingale condition is
\[
        (1-a)\left(\frac23 b+2c\right)+a=1,
\]
which, for \(a<1\), is equivalent to \((2/3)b+2c=1\).  Together with \(b+c=1\), this forces \(b=3/4\) and \(c=1/4\), giving \eqref{eq:X_family}.

First minimize \(\KL(Q^Y(a)\Vert P)\).  More generally, for
\[
        Q^{b,c}(a)=\bigl((1-a)b,\,a,\,(1-a)c\bigr),
        \qquad b+c=1,
\]
one has
\[
        \KL(Q^{b,c}(a)\Vert P)
        =a\log\frac{a}{p_2}
        +(1-a)b\log\frac{(1-a)b}{p_1}
        +(1-a)c\log\frac{(1-a)c}{p_3}.
\]
Moreover,
\[
        \frac{d^2}{da^2}\KL(Q^{b,c}(a)\Vert P)
        =\frac1a+\frac1{1-a}>0,
\]
so the objective is strictly convex and the first-order condition gives the unique minimizer:
\begin{equation}\label{eq:reverse_min_general}
        \frac{a}{1-a}
        =p_2\frac{b^bc^c}{p_1^bp_3^c}.
\end{equation}
For the \(Y\)-numeraire family, \(b=c=1/2\).  Therefore
\[
        a_Y^\dagger
        =\argmin_a\KL(Q^Y(a)\Vert P)
        =\frac{\sqrt6}{12+\sqrt6}
        \approx 0.169521.
\]
For the \(X\)-numeraire family, \(b=3/4\) and \(c=1/4\).  Therefore
\[
        a_X^\dagger
        =\argmin_a\KL(Q^X(a)\Vert P)
        =\frac{2^{3/4}}{8+2^{3/4}}
        \approx 0.173707.
\]
Since \(a_X^\dagger\neq a_Y^\dagger\),
\[
        T_RQ^Y(a_Y^\dagger)=Q^X(a_Y^\dagger)
        \neq
        Q^X(a_X^\dagger).
\]
Thus the reverse-entropy projection computed in \(Y\)-units is not carried by the likelihood transform to the reverse-entropy projection computed directly in \(X\)-units.\end{proof}

\begin{corollary}[A digital-option price disagreement between the two MEMMs]\label{cor:memm_price_disagreement}
Let
\[
 Q^{Y,\dagger}:=Q^Y(a_Y^\dagger),
 \qquad
 Q^{X,\dagger}:=Q^X(a_X^\dagger),
\]
be the MEMMs selected independently in the two numeraires, and consider the digital option
\[
 H_T:=Y_T\one_{\{\omega_2\}},
\]
which pays one unit of the \(Y\)-numeraire in the middle state. The two selected pricing rules give
\begin{align}
 \pi_Y^\dagger(H)
 &=Y_0E_{Q^{Y,\dagger}}\!\left[\frac{H_T}{Y_T}\right]
 =Y_0a_Y^\dagger,\label{eq:Y_memm_digital_price}\\
 \pi_X^\dagger(H)
 &=X_0E_{Q^{X,\dagger}}\!\left[\frac{H_T}{X_T}\right]
 =Y_0a_X^\dagger.\label{eq:X_memm_digital_price}
\end{align}
Since \(a_Y^\dagger\ne a_X^\dagger\), the two agents assign different prices to the same claim.
\end{corollary}

\begin{proof}
The first equality follows from \(Q^{Y,\dagger}(\omega_2)=a_Y^\dagger\). Since
\[
 \frac{Y_T}{X_T}=\frac{Y_0}{X_0}\frac1R
\]
and \(R(\omega_2)=1\), the \(X\)-based price is
\[
 X_0E_{Q^{X,\dagger}}\!\left[\frac{Y_T}{X_T}\one_{\{\omega_2\}}\right]
 =Y_0E_{Q^{X,\dagger}}\!\left[\frac1R\one_{\{\omega_2\}}\right]
 =Y_0a_X^\dagger.
\]
The payoff vector \((0,1,0)\) is not in the span of the traded payoff vectors \(\one\) and \(R\), so the claim is nonreplicable.
\end{proof}

\begin{remark}
The discrepancy in Theorem~\ref{thm:reverse_counterexample} is structural, not numerical. Equation~\eqref{eq:reverse_difference_general} shows that a nontrivial likelihood transform changes the reverse objective by a candidate-dependent term. Proposition~\ref{prop:pricing_consistency} then turns the moved argmin into a pricing consequence: the two MEMMs agree on the traded assets and on every replicable claim, but they need not define the same completion price for a nonreplicable claim. In contrast, the forward selectors are likelihood compatible and therefore generate the same pricing functional in both units.
\end{remark}

\section{A trinomial illustration of exact forward completion}\label{sec:trinomial}

The generic one-parameter family used in the preceding proof also makes the forward calculation transparent. For
\[
 Q^{b,c}(a)=\bigl((1-a)b,\,a,\,(1-a)c\bigr),
 \qquad 0<a<1,\quad b,c>0,\quad b+c=1,
\]
the forward divergence has the form
\begin{equation}\label{eq:generic_forward_trinomial}
 \KL(P\Vert Q^{b,c}(a))
 =C_{b,c}(P)-(1-p_2)\log(1-a)-p_2\log a,
\end{equation}
where \(C_{b,c}(P)\) does not depend on \(a\). Its unique minimizer is therefore
\begin{equation}\label{eq:generic_forward_a}
 a^*=p_2,
\end{equation}
independently of \(b\) and \(c\). In the trinomial transformation used above, the change of numeraire alters the conditional allocation \((b,c)\) of the remaining mass but not the optimal coordinate \(a^*\). This is the coordinate-level version of the general covariance result in that family.

For the numerical model of Theorem~\ref{thm:reverse_counterexample},
\[
        R=\left(\frac32,1,\frac12\right),
        \qquad
        P=\left(\frac12,\frac16,\frac13\right),
\]
the \(Y\)-numeraire family is \eqref{eq:Y_family}. Equation~\eqref{eq:generic_forward_a} gives
\[
        a^*=p_2=\frac16,
        \qquad
        Q^{Y,*}=\left(\frac5{12},\frac16,\frac5{12}\right).
\]
The corresponding \(X\)-numeraire measure is
\[
        Q^{X,*}=RQ^{Y,*}
        =\left(\frac58,\frac16,\frac5{24}\right),
\]
and, by Theorem~\ref{thm:forward_invariance}, it is also the forward projection in \(X\)-units.

A portfolio investing fraction \(w\) in \(X\) and \(1-w\) in \(Y\) has gross return
\[
        R^w=1+w(R-1)
        =\left(1+\frac w2,\,1,\,1-\frac w2\right).
\]
The expected log-growth objective is
\[
        g(w)=\frac12\log\left(1+\frac w2\right)
              +\frac13\log\left(1-\frac w2\right),
        \qquad -2<w<2,
\]
whose maximizer is \(w^*=2/5\).  Hence
\[
        R^{w^*}=\left(\frac65,1,\frac45\right).
\]
On the other hand,
\[
        \frac{P}{Q^{Y,*}}
        =\left(\frac65,1,\frac45\right)
        =R^{w^*}.
\]
Therefore
\[
        R^{w^*}Q^{Y,*}=P,
        \qquad
        \KL(P\Vert R^{w^*}Q^{Y,*})=0.
\]
This zero residual does not mean the market is complete. As Lemma~\ref{lem:finite_existence} and Theorem~\ref{thm:forward_log_wealth} show in the finite-state interior case, it means that the selected forward-entropy state-price completion is the one under which the log-optimal belief payoff \(dP/dQ^{Y,*}\) lies in the span of the traded assets. By Proposition~\ref{prop:pricing_consistency}, the \(Y\)- and \(X\)-based forward selectors agree on the value of every integrable claim. In particular, for the digital claim in Corollary~\ref{cor:memm_price_disagreement}, both give the price
\[
 Y_0Q^{Y,*}(\omega_2)=\frac{Y_0}{6}.
\]

\section{Scope beyond finite states}\label{sec:scope}

Theorem~\ref{thm:forward_invariance} is not a finite-state artifact. It is the elementary identity
\[
        \KL(P\Vert LQ)=\KL(P\Vert Q)-E_P\log L,
\]
and therefore applies on any filtered probability space whenever the likelihood transform \(Q\mapsto LQ\) is defined and maps the admissible class for one numeraire bijectively to the admissible class for another. In semimartingale models, this is precisely the role of the change-of-numeraire theorem under its usual true-martingale hypotheses.

The finite-state restriction enters in Lemma~\ref{lem:finite_existence} and Theorem~\ref{thm:forward_log_wealth}, where existence of an equivalent minimizer and a classical Lagrange multiplier argument are automatic under the strictly positive feasible-point condition. In continuous-time incomplete markets the forward minimizer over equivalent martingale measures may fail to exist as an equivalent martingale measure. The logarithmic utility dual optimizer is naturally formulated over a larger domain of supermartingale deflators or optional measures. This is one reason the MEMM remains especially useful: under suitable finite-entropy assumptions it yields an equivalent martingale measure and is tied to exponential utility.

Thus the results here should be read as follows. On any martingale-measure domain stable under change of numeraire, the forward entropy projection is numeraire invariant and its selected pricing functional is independent of the accounting representation. In finite-state markets, the projection exists under a transparent interiority condition, its likelihood ratio is exactly the log-optimal payoff, and its divergence is the sharp upper bound on expected log growth. In more general semimartingale models, the same constant-shift mechanism points to a formulation on paired deflator domains rather than only on equivalent martingale measures. The present paper identifies the identity and pricing covariance that such an extension would have to preserve.

\section{Conclusion}\label{sec:conclusion}

The two relative-entropy orientations solve different optimization problems. The convex-conjugate calculation associates \(\KL(P\Vert Q)\) with logarithmic utility and \(\KL(Q\Vert P)\) with exponential utility. On the primal side, logarithmic utility changes by an additive portfolio-independent term under a change of numeraire, while exponential utility specified in a fixed accounting unit does not. There is therefore no unconditional ranking of the two measures: each is natural for its corresponding optimization problem.

For a martingale-measure selector intended to provide a canonical completion of an incomplete market, however, numeraire covariance is economically substantive. It is equivalent to obtaining the same claim-pricing functional in every numeraire. The forward projection has this property, so agents computing it in \(N\)- and \(X\)-units obtain likelihood-compatible measures and agree on every integrable claim. The MEMM need not: in the trinomial example its two independently computed versions are incompatible and assign different prices to a nonreplicable digital claim, although they still agree on all traded and replicable payoffs.

Finally, the logarithmic orientation is not merely one convenient invariant example. Within smooth \(f\)-divergence projections, requiring the invariance criterion uniformly over the elementary finite one-period likelihood-ratio families singles it out, up to affine equivalence and positive scaling. Together with the exact log-growth duality, this identifies the forward entropy projection as the natural choice when representation-invariant completion, rather than fixed-unit exponential-utility valuation, is the objective.

\section*{Statements and Declarations}

\noindent\textbf{Funding.} This work was partially supported by the Grant Agency of the Czech Republic under grant 24-11146S.

\noindent\textbf{Competing interests.} The author declares no competing interests.

\noindent\textbf{Data availability.} Data sharing is not applicable to this article as no datasets were generated or analyzed.

\noindent\textbf{Code availability.} Not applicable.

\noindent\textbf{Use of AI tools.} The author used ChatGPT 5.6 for language editing, organization, and suggestions concerning mathematical exposition.

\end{document}